\documentclass[12pt]{article}

\usepackage[margin=1in]{geometry}
\usepackage{CJKutf8, amsmath, amssymb, amsthm, authblk, hyperref, enumitem}
\usepackage[basic]{complexity}
\usepackage{wasysym}
\usepackage{cleveref}
\usepackage{physics}

\newtheorem{lemma}{Lemma}
\newtheorem{theorem}{Theorem}
\newtheorem{corollary}{Corollary}

\theoremstyle{definition}
\newtheorem{definition}{Definition}
\newtheorem{assumption}{Assumption}
\newtheorem{question}{Question}

\theoremstyle{remark}
\newtheorem*{remark}{Remark}

\DeclareMathOperator{\diag}{diag}
\DeclareMathOperator*{\e}{\mathbb E}
\def\be#1\ee{\begin{equation}#1\end{equation}}
\def\ba#1\ea{\begin{align}#1\end{align}}
\def\bas#1\eas{\begin{align*}#1\end{align*}}

\usepackage[style=trad-unsrt, citestyle=numeric-comp, giveninits=true, doi=false, url=false, eprint=true, isbn=false]{biblatex}

\begin{document}

\title{Thermalization without eigenstate thermalization}

\author[1]{Aram W. Harrow\thanks{aram@mit.edu}}
\author[1,2]{Yichen Huang (黄溢辰)\thanks{yichenhuang@fas.harvard.edu}}
\affil[1]{Center for Theoretical Physics, Massachusetts Institute of Technology, Cambridge, Massachusetts 02139, USA}
\affil[2]{Department of Physics, Harvard University, Cambridge, Massachusetts 02138, USA}

\begin{CJK}{UTF8}{gbsn}

\maketitle

\end{CJK}

\begin{abstract}

In an isolated quantum many-body system undergoing unitary evolution, we study the thermalization of a subsystem, treating the rest of the system as a bath. In this setting, the eigenstate thermalization hypothesis (ETH) was proposed to explain thermalization. Consider a nearly integrable Sachdev-Ye-Kitaev model obtained by adding random all-to-all $4$-body interactions as a perturbation to a random free-fermion model. When the subsystem size is larger than the square root of but is still a vanishing fraction of the system size, we prove thermalization if the system is initialized in a random product state, while almost all eigenstates violate the ETH. In this sense, the ETH is not a necessary condition for thermalization.

\end{abstract}

Preprint number: MIT-CTP/5467

\section{Introduction}

\subsection{Background}

Thermalization is a fundamental process in nature. It says that a system in contact with a bath tends to evolve to a Gibbs state described by the canonical ensemble. Suppose an isolated quantum many-body system is initialized in a pure state. Under unitary evolution the system stays in a pure state and never thermalizes. To study thermalization, we divide the system into two parts $A$ and $\bar A$ such that subsystem $A$ is much smaller than $\bar A$. We view $\bar A$ as a bath of $A$ and consider the thermalization of $A$, i.e., whether physical properties measured on $A$ evolve to those of a Gibbs state.

An important goal of statistical mechanics is to understand the mechanism of thermalization. One proposal is the eigenstate thermalization hypothesis (ETH) \cite{Deu91, Sre94, RDO08, DLL18}.

\begin{definition}
A state is $A$-thermal if its reduced density matrix for subsystem $A$ is approximately equal to\footnote{More precisely, ``approximately equal to'' means that the trace distance between the two reduced density matrices vanishes in the thermodynamic limit.} that of a Gibbs state with the same energy.
\end{definition}

\begin{definition} [eigenstate thermalization hypothesis] \label{def:eigt}
An eigenstate obeys the ETH with respect to subsystem $A$ if it is $A$-thermal.
\end{definition}

An important feature of this definition is the size of subsystem $A$. If a state is $A$-thermal for all subsystems $A$ of size $L$, then its $L$-point correlation functions are approximately equal to those of the Gibbs state. In the literature, the ETH is often defined only for $L=O(1)$. However, even some simple and experimentally accessible observables require considering $L$'s that grow with the system size $N$.  The scaling of $L$ with $N$ will be discussed in the context of the ETH in Subsection \ref{sec:ETH}.

Definition \ref{def:eigt} refers only to individual eigenstates, but one might also want to discuss the ETH for systems (i.e., Hamiltonians). In the literature, the term ``strong ETH'' or ``weak ETH'' \cite{BKL10, KLW15, BCSB19}, respectively, refers to systems where all or almost all eigenstates (possibly within some energy interval) obey the ETH. In this paper, unless otherwise noted, statements about the ETH refer to individual eigenstates.

For a Hamiltonian $H$, let $|\Psi_0\rangle$ and $|\Psi_t\rangle:=e^{-iHt}|\Psi_0\rangle$ be the initial and time-evolved states, respectively.

\begin{definition} [thermalization] \label{def:t}
Thermalization means that $|\Psi_t\rangle$ becomes $A$-thermal as $t$ grows.
\end{definition}

We present a well-known argument \cite{GE16, DKPR16, Deu18} for ``ETH implies thermalization'' based on two assumptions.

\begin{assumption} \label{ass:1}
The spectrum of $H$ is non-degenerate (all eigenvalues are distinct).
\end{assumption}

\begin{remark}
Intuitively, this assumption is usually valid if $H$ does not have any symmetry. Rigorously, it holds with probability $1$ if $H$ is a random Hamiltonian from the Gaussian unitary ensemble. It also holds for almost every local Hamiltonian on a lattice \cite{KLW15}.
\end{remark}

Let $|1\rangle,|2\rangle,\ldots$ be a complete set of eigenstates of $H$ with corresponding energies $E_1,E_2,\ldots$.

\begin{assumption} \label{ass:2}
The energy distribution of $|\Psi_0\rangle$ is sharply peaked around the mean $E:=\langle\Psi_0|H|\Psi_0\rangle$ in the sense that
\begin{equation} \label{eq:ass2}
\sum_{j:~E_j\approx E}p_j\approx1,\quad p_j:=\big|\langle j|\Psi_0\rangle\big|^2.
\end{equation}
\end{assumption}

\begin{remark}
If $H$ is a local Hamiltonian on a lattice, (\ref{eq:ass2}) with explicit error bounds was proved for any $|\Psi_0\rangle$ with exponential decay of correlations \cite{Ans16}.
\end{remark}

Assumption \ref{ass:1} implies that the time-averaged state
\begin{equation}
\bar\Psi:=\lim_{\tau\to\infty}\frac1\tau\int_0^\tau|\Psi_t\rangle\langle\Psi_t|\,\mathrm dt=\sum_jp_j|j\rangle\langle j|
\end{equation}
is obtained by dephasing $|\Psi_0\rangle$ in the energy eigenbasis. Then, Assumption \ref{ass:2} implies that
\begin{equation} \label{eq:expl}
\bar\Psi\approx\sum_{j:~E_j\approx E}p_j|j\rangle\langle j|.
\end{equation}
If every $|j\rangle$ in the sum on the right-hand side is $A$-thermal, then $\bar\Psi$ is $A$-thermal.

To establish thermalization, it is necessary that $\bar\Psi$ is $A$-thermal. Furthermore, one needs to prove equilibration, i.e., the temporal fluctuation of the reduced density matrix of $|\Psi_t\rangle$ for subsystem $A$ is small. This can be done under mild additional assumptions \cite{Rei08, LPSW09, Sho11}.

\begin{question} \label{q}
Is the ETH a necessary condition for thermalization?
\end{question}

A positive answer to this question would further justify the essence of the ETH as an explanation for the emergence of the canonical ensemble from unitary evolution. If the answer is negative, then it is time to call for other mechanisms of thermalization. Either way Question \ref{q} is illuminating.

The answer to Question \ref{q} depends on the set of initial states under consideration. If the initial state is an eigenstate, then the system does not evolve and thus thermalization trivially implies the ETH. However, eigenstates of local Hamiltonians typically have very high complexity and cannot be efficiently prepared; they could also be considered ``fine tuned'' since they form a discrete set.

De Palma et al.~\cite{DSGC15} considered initial states of the form $|\Psi_A\rangle\otimes|\Psi_{\bar A}\rangle$, where $|\Psi_A\rangle$ is an arbitrary pure state of subsystem $A$; $|\Psi_{\bar A}\rangle$ has a sharply peaked energy distribution but is otherwise arbitrary. If {\em all} such initial states thermalize, the ETH was proved under some assumptions, one of which is that the Hilbert space dimension of $A$ is much smaller than the heat capacity of $\bar A$. In a system of $N$ qubits, if the heat capacity is extensive, this assumption implies that $L\lesssim\ln N$, where $L$ is the number of qubits in $A$.  This means that almost all qubits are in $\bar A$.  Since the states $\ket{\Psi_{\bar A}}$ are so general, they typically have very high complexity.

The results of Ref.~\cite{DSGC15} can be rephrased as saying that a system violating the ETH must fail to thermalize for at least one initial state of the form $\ket{\Psi_A}\otimes \ket{\Psi_{\bar A}}$, where the size of $A$ is small and $\ket{\Psi_{\bar A}}$ has a sharply peaked energy distribution.  This leaves open the question of whether typical low-complexity initial states thermalize.

We will analyze initial states that are product across {\em all} cuts, not only the cut between $A$ and $\bar A$.  This choice is because these states are more relevant to experiments and are more plausible models of naturally occurring states.  

\subsection{Results (informal)} \label{sec:summary}

In this paper, we consider a nearly integrable complex Sachdev-Ye-Kitaev (SYK) model in a system of $N$ (Dirac) fermionic modes. The model has fermion number conservation and is obtained by adding random all-to-all $4$-body interactions as a perturbation to a random free-fermion model. If the perturbation is sufficiently small, the eigenstates are close to random Gaussian states with definite fermion number. We prove that they obey and violate the ETH with overwhelming probability for $L\ll \sqrt N$ and $L\gtrsim\sqrt N$, respectively, where $L$ is the number of fermionic modes in subsystem $A$. (We write $x\ll y$ if $x/y\to 0$ as $N\to\infty$; $x \gtrsim y$ if $x/y\ge c$ for some constant $c>0$.) Previously, Mag\'an \cite{Mag16} showed that random Gaussian states with definite fermion number obey the ETH in an average sense for $L\le2$.

For $L\gtrsim\sqrt N$, since the ETH fails, there is an observable (on subsystem $A$) that distinguishes an eigenstate from the thermal state.  Measuring this observable does not require many-body entangling operations. It can be done by measuring the occupation numbers of $L$ individual modes, in a basis chosen based on the eigenstate in question, and then classically processing the measurement results.  We describe this in more detail in Subsection \ref{ss:met}.

Let the initial state $|\Psi_0\rangle$ be a (random) product state, where each fermionic mode is either vacant or occupied. No matter how small the perturbation is, its effect on the dynamics becomes significant at sufficiently long times when the time is greater than the inverse of the perturbation strength. Previously, we proved entanglement thermalization with high probability \cite{HH23}.

\begin{definition} [entanglement thermalization \cite{ZKH15}] \label{def:entt}
For $L\le N/2$ and to leading order in $L$, the entanglement entropy of subsystem $A$ evolves to the thermodynamic entropy of $A$ at the same energy.
\end{definition}

Here we prove thermalization (Definition \ref{def:t}) with high probability when $L \ll N/\ln N$. However, thermalization never occurs when $L\ge cN$ for an arbitrarily small constant $c>0$. 

Table \ref{t:summary} summarizes our results. ``Thermalization without eigenstate thermalization'' (the title of this paper) is proved for $\sqrt N \lesssim L \ll N/\ln N$. Although the subsystem size $L$ is not upper bounded by a constant, it is still a vanishing fraction of the system size $N$.

\begin{table}
\caption{Summary of results in the thermodynamic limit $N\to \infty$. $L$ is the subsystem size. While eigenstate thermalization is a static property of the Hamiltonian, thermalization and entanglement thermalization are dynamic processes, in which the initial state is a (random) product state. Smiley (frown) means that the phenomenon in the column occurs (does not occur) when $N$ and $L$ satisfy the relation in the row. ``Thermalization without eigenstate thermalization'' is proved for $\sqrt N \lesssim L \ll N/\ln N$.}
\centering
\begin{tabular}{|c|c|c|c|} 
\hline
& eigenstate & thermalization & entanglement \\
& thermalization & & thermalization \\
& (Definition \ref{def:eigt}) & (Definition \ref{def:t}) & (Definition \ref{def:entt}) \\
\hline
$L \ll \sqrt N$ & {\Large\smiley} Theorem \ref{thm:ethv} & {\Large\smiley} Theorem \ref{cor} & \\
\cline{1-2}
$\sqrt N \lesssim L \ll N/\ln N$ & & & {\Large\smiley} \cite{HH23} \\
\cline{1-1} \cline{3-3} 
  $cN < L\le N/2$ & {\Large\frownie} Theorems \ref{thm:non}, \ref{thm:nontd} & {\Large\frownie} Theorem \ref{thm:nt} & \\
  for arbitrarily small constant  $c>0$  &&&\\
\hline
\end{tabular}
\label{t:summary}
\end{table}

``Thermalization without eigenstate thermalization'' can be understood as follows. $\bar\Psi$ is $A$-thermal if every $|j\rangle$ in the sum on the right-hand side of (\ref{eq:expl}) is $A$-thermal. However, this gives only a sufficient condition for $\bar\Psi$ being $A$-thermal. It could be possible that while most $|j\rangle$'s are not $A$-thermal, $\bar\Psi$ (a mixture of many $|j\rangle$'s) is $A$-thermal because the deviations of different $|j\rangle$'s from $A$-thermality cancel. This possibility provably occurs in our model.

\subsection{Eigenstate thermalization hypothesis} \label{sec:ETH}

\paragraph{ETH for systems.}Definition \ref{def:eigt} is the definition of the ETH for individual eigenstates. Previous work has also defined the ETH for systems: A Hamiltonian obeys the strong or weak ETH (in an energy interval) if all or almost all\footnote{``Almost all'' means that the fraction of ETH-violating eigenstates (in the energy interval) vanishes in the thermodynamic limit.} eigenstates (in the energy interval) obey the ETH, respectively.

Our results on the ETH (summarized in the second column of Table \ref{t:summary}) apply to random eigenstates of random Hamiltonians.  If we view them as statements about Hamiltonians, they imply that the weak ETH holds and fails with overwhelming probability for $L \ll \sqrt{N}$ and $L\gtrsim\sqrt{N}$, respectively.

Previous work by Mori and Shiraishi~\cite{MS17} showed evidence of thermalization in a model that obeys the weak ETH but not the strong ETH.\footnote{Using a combination of analytical and numerical methods, it was shown \cite{SM17} that an exponentially small (in the system size) fraction of the eigenstates of the model violate the ETH.} In their model, for a generic initial state, the total weight of ETH-violating eigenstates in the sum on the right-hand side of (\ref{eq:expl}) is negligible.  Thus, the thermalization observed in Ref. \cite{MS17} could be explained by the weak ETH. By contrast, in our model, for $\sqrt N \lesssim L \ll N/\ln N$, we observe thermalization without even the weak ETH, so $\bar\Psi$ is $A$-thermal even though most eigenstates on the right-hand side of (\ref{eq:expl}) are not.

\paragraph{Subsystem size in ETH.}The idea behind the ETH is that eigenstates should look thermal with respect to ``simple'' observables.  Since we lack a general proof of when the ETH holds, we cannot precisely determine which observables should be included here.  One plausible approach is to consider all observables that act non-trivially only on a sufficiently small subsystem. Indeed, Ref.~\cite{GG18} provided evidence that in systems with spatially local interactions, the ETH fails if the subsystem size is a constant fraction of the system size.

In our model, we observe a sharp threshold in the subsystem size $L$: If $L\ll \sqrt{N}$ then the ETH holds for almost all eigenstates and if $L \gtrsim\sqrt{N}$ then it fails for almost all eigenstates. Our model (\ref{eq:model}) has non-local interactions and has the same set of eigenstates as an integrable Hamiltonian. In the future, it would be interesting to study the validity of the ETH with respect to the subsystem size in systems that are non-integrable and/or have spatially local interactions.

Higher values of $L$ are relevant for more complicated observables such as $L$-point correlation functions.  They can also control the probability of large fluctuations for simpler observables~\cite{Touchette09}.  For example, in an $N$-mode fermionic system, let $Q$ be the total fermionic number operator (\ref{eq:Q}).  For $k\le N$, the expectation $\expval{Q^k}$ depends on $k$-point correlation functions.  While $k=1$ and $k=2$ give the expectation and variance of $Q$, higher values of $k$ can yield sharper bounds on the probability of large fluctuations in $Q$.

The scaling of $L$ with $N$ can be interpreted as the important question of how large the bath needs to be.  In an isolated system of size $N$, (eigenstate) thermalization with respect to subsystem $A$ of size $L$ means that a bath $\bar A$ of size $N-L$ suffices.  In our model, the threshold for the validity of the ETH is $L\sim\sqrt N$. Thus, the bath size must be $\gg L^2$ for the ETH to hold.  This is unrealistic when $L\sim 10^{23}$ is macroscopically large.  By contrast, thermalization occurs when $L\ll N/\ln N$ or as long as the bath size is $\gg L\ln L$. The ratio of the bath size to $L$ is not huge even if $L\sim 10^{23}$.

\section{Results (formal)}

\subsection{Model and definitions} \label{ss:MD}

The real (complex) SYK$q$ model \cite{SY93, Kit15, Sac15, MS16, GKST20} is a quantum mechanical model of Majorana (Dirac) fermions with random all-to-all $q$-body interactions (``$q$-body'' means that each term in the Hamiltonian acts non-trivially only on $q$ fermionic modes).

Consider an $N$-mode (Dirac) fermionic system with creation and annihilation operators $a_j^\dag,a_j$ indexed by $j=1,2,\ldots,N$. Let $A$ be an arbitrary subsystem of $L$ fermionic modes and $\bar A$ be the complement of $A$ (rest of the system).

\begin{definition} [complex SYK2 model]
Let $h$ be a random matrix of order $N$ from the Gaussian unitary ensemble. The Hamiltonian of the complex SYK2 model is
\begin{equation}
H_\textnormal{SYK2}=a^\dag ha,
\end{equation}
where $a:=(a_1,a_2,\ldots,a_N)^T$ is a column vector of $N$ annihilation operators.
\end{definition}

\begin{definition} [complex SYK4 model \cite{Sac15, GKST20}]
Let
\begin{equation} 
\mathcal I:=\left\{(j,k,l,m)\in\mathbb\{1,2,\ldots,N\}^{\times4}:(j<k)~\textnormal{and}~(l<m)~\textnormal{and}~(jN+k\le lN+m)\right\}
\end{equation}
and $J:=\{J_{jklm}\}_{(j,k,l,m)\in\mathcal I}$ be a collection of $|\mathcal I|$ independent complex Gaussian random variables with zero mean $\overline{J_{jklm}}=0$ and unit variance $\overline{|J_{jklm}|^2}=1$. The Hamiltonian of the complex SYK4 model is
\begin{equation}
H_\textnormal{SYK4}=\sum_{(j,k,l,m)\in\mathcal I}J_{jklm}a_j^\dag a_k^\dag a_la_m+\textnormal{H.c.},
\end{equation}
where ``H.c.'' means Hermitian conjugate.
\end{definition}

The complex SYK$q$ model is also known as the embedded Gaussian unitary ensemble \cite{Kot01, BW03} and has been studied under this name for decades.

Let
\begin{equation} \label{eq:Q}
Q:=\sum_{j=1}^Na_j^\dag a_j
\end{equation}
be the fermion number operator. Let $\epsilon_1,\epsilon_2$ be infinitesimal and
\begin{equation}
H_\textnormal{SYK}:=H_\textnormal{SYK2}+\epsilon_2H_\textnormal{SYK4}.
\end{equation}
Our model is
\begin{equation} \label{eq:model}
  H
  =Q+\epsilon_1H_\textnormal{SYK}
  =Q+\epsilon_1H_\textnormal{SYK2}+\epsilon_1\epsilon_2H_\textnormal{SYK4}.
\end{equation}
Both $H_\textnormal{SYK}$ and $H$ are nearly integrable as both $H_\textnormal{SYK2}$ and $Q+\epsilon_1H_\textnormal{SYK2}$ are integrable models of free fermions. By definition, the complex SYK2 and SYK4 models and hence $H_\textnormal{SYK}$ and $H$ conserve fermion number in that
\begin{equation} \label{eq:comm}
[H_\textnormal{SYK2},Q]=[H_\textnormal{SYK4},Q]=[H_\textnormal{SYK},Q]=[H,Q]=0.
\end{equation}

Let
\begin{equation}
\sigma_\beta:=e^{-\beta H}/\tr(e^{-\beta H})
\end{equation}
be a thermal state at inverse temperature $\beta$. Neglecting infinitesimal quantities, the reduced density matrix of subsystem $A$ is
\begin{equation} \label{eq:thr}
\sigma_{\beta,A}:=\tr_{\bar A}\sigma_\beta=e^{-\beta Q_A}/\tr(e^{-\beta Q_A}),
\end{equation}
where
\begin{equation} \label{eq:qa}
Q_A:=\sum_{j\in A}a_j^\dag a_j
\end{equation}
is the restriction of $Q$ to $A$. Let $|\Phi\rangle$ be such that that $Q|\Phi\rangle=n|\Phi\rangle$. If $\sigma_\beta$ and $|\Phi\rangle$ have the same energy, then
\begin{equation} \label{eq:temp}
\tr(\sigma_\beta H)=\langle\Phi|H|\Phi\rangle=n\implies\beta=\ln(N/n-1).
\end{equation}

Let $\|B\|_1:=\tr\sqrt{B^\dag B}$ denote the trace norm of a linear operator $B$. The trace distance
\begin{equation} \label{eq:tr}
T(\rho_1,\rho_2):=\|\rho_1-\rho_2\|_1/2=\max_{\|B\|\le1}|\tr(\rho_1B)-\tr(\rho_2B)|/2,\quad0\le T(\rho_1,\rho_2)\le1
\end{equation}
and the fidelity
\begin{equation}
F(\rho_1,\rho_2):=\tr^2\sqrt{\sqrt{\rho_1}\rho_2\sqrt{\rho_1}},\quad0\le F(\rho_1,\rho_2)\le1
\end{equation}
are measures of distinguishability between two density operators. It is well known that
\begin{equation} \label{eq:rel}
1-\sqrt{F(\rho_1,\rho_2)}\le T(\rho_1,\rho_2)\le\sqrt{1-F(\rho_1,\rho_2)}.
\end{equation}

The trace distance is directly related to the success probability of the optimal protocol for distinguishing two states \cite[Chap. 9]{NC00}. Specifically, consider the following state inference problem. We are given a random state $\rho$, which is either $\rho_1$ or $\rho_2$ with equal probability. We are allowed to perform a measurement on a single copy of $\rho$. From the measurement results, we must predict whether $\rho$ is $\rho_1$ or $\rho_2$. Let
\begin{equation}
\rho_1-\rho_2=\sum_j\mu_j|\Phi_j\rangle\langle\Phi_j|\quad\textnormal{with}\quad\sum_j\mu_j=0
\end{equation}
be the eigendecomposition of $\rho_1-\rho_2$. The optimal protocol is to measure $\rho$ in the orthonormal basis $\{|\Phi_j\rangle\}$. If the post-measurement state is $|\Phi_j\rangle$, then we predict $\rho_1$ if $\mu_j\ge0$ and predict $\rho_2$ if $\mu_j<0$. The success probability of this protocol is
\begin{equation}
\frac12+\frac12\sum_j|\mu_j|=\frac{1+T(\rho_1,\rho_2)}2.
\end{equation}

We use standard asymptotic notation. Let $f,g:\mathbb R^+\to\mathbb R^+$ be two functions. One writes $f(x)=O(g(x))$ if and only if there exist constants $M,x_0>0$ such that $f(x)\le Mg(x)$ for all $x>x_0$; $f(x)=\Omega(g(x))$ if and only if there exist constants $M,x_0>0$ such that $f(x)\ge Mg(x)$ for all $x>x_0$; $f(x)=\Theta(g(x))$ if and only if there exist constants $M_1,M_2,x_0>0$ such that $M_1g(x)\le f(x)\le M_2g(x)$ for all $x>x_0$; $f(x)=o(g(x))$ if and only if for any constant $M>0$ there exists a constant $x_0>0$ such that $f(x)<Mg(x)$ for all $x>x_0$.

\subsection{Eigenstate thermalization}

The spectrum of $H_\textnormal{SYK2}$ is non-degenerate with probability 1 \cite{HH23}. Then, due to fermion number conservation (\ref{eq:comm}) and since the perturbation $\epsilon_2H_\textnormal{SYK4}$ is infinitesimal, $H$, $H_\textnormal{SYK}$, $H_\textnormal{SYK2}$, and $Q+\epsilon_1H_\textnormal{SYK2}$ have the same set of eigenstates (up to an infinitesimal error), each of which has a definite fermion number.

Let $|\psi\rangle$ be an eigenstate of $H_\textnormal{SYK2}$ with fermion number $n$ and $\nu:=n/N$ be the filling fraction. We write $|\psi\rangle\sim\mathcal P_{N,n}$ if $|\psi\rangle$ is randomly sampled using the following procedure:
\begin{enumerate} [nosep]
\item Let $h$ be a random matrix of order $N$ from the Gaussian unitary ensemble.
\item $|\psi\rangle$ is chosen uniformly at random from the ${N\choose n}$ eigenstates of $a^\dag ha$ with fermion number $n$.
\end{enumerate}
Let $\psi_A:=\tr_{\bar A}|\psi\rangle\langle\psi|$ be the reduced density matrix. Recall the definition (\ref{eq:thr}) of $\sigma_{\beta,A}$, where $\beta$ is given by (\ref{eq:temp}).

\begin{theorem} [eigenstate thermalization] \label{thm:ethv}
Suppose that $1/2\ge\nu=\Omega(1)$. For $L=o(\sqrt N)$ and any $\Delta$ such that $\Omega(L^2)=\Delta=o(N)$,
\begin{gather}
\Pr_{|\psi\rangle\sim\mathcal P_{N,n}}\big(F(\psi_A,\sigma_{\beta,A})=1-O(\Delta/N)\big)=1-O(e^{-\Delta}),\label{eq:ethf}\\
\Pr_{|\psi\rangle\sim\mathcal P_{N,n}}\big(T(\psi_A,\sigma_{\beta,A})=O(\sqrt{\Delta/N})\big)=1-O(e^{-\Delta}).\label{eq:etht}
\end{gather}
\end{theorem}

\begin{theorem} [failure of ETH] \label{thm:non}
Suppose that $1/2\ge\nu=\Omega(1)$. For $L=\Omega(\sqrt{N})$,
\begin{gather}
\Pr_{|\psi\rangle\sim\mathcal P_{N,n}}\big(F(\psi_A,\sigma_{\beta,A})=e^{-\Omega(L^2/N)}\big)=1-e^{-\Omega(L^2)},\label{eq:nons}\\
\Pr_{|\psi\rangle\sim\mathcal P_{N,n}}\big(T(\psi_A,\sigma_{\beta,A})=1-e^{-\Omega(L^2/N)}\big)=1-e^{-\Omega(L^2)}.\label{eq:nont}
\end{gather}
\end{theorem}

\begin{theorem} [failure of ETH] \label{thm:nontd}
Suppose that $\nu\le1/2$. For $L>n$,
\begin{equation}
F(\psi_A,\sigma_{\beta,A})\le(1-\nu)^{L-n},\quad T(\psi_A,\sigma_{\beta,A})\ge1-(1-\nu)^{L-n}.
\end{equation}
\end{theorem}

These results on the ETH are for individual eigenstates. To interpret Theorems \ref{thm:ethv}, \ref{thm:non} as statements about Hamiltonians, recall the definition of the weak ETH in an energy interval in Subsection \ref{sec:ETH}.

\begin{corollary}
Let $n$ be a positive integer such that $N/2\ge n=\Omega(N)$. For $L=o(\sqrt N)$, the probability (with respect to the randomness of $h$) that $Q+\epsilon_1a^\dag ha$ obeys the weak ETH in the energy interval $(n-1/2,n+1/2)$ is $1-e^{-\Delta}$ for any $\Delta=o(N)$. For $L=\Omega(\sqrt N)$, the probability that $Q+\epsilon_1a^\dag ha$ obeys the weak ETH in the energy interval $(n-1/2,n+1/2)$ is $e^{-\Omega(L^2)}$.
\end{corollary}

Near the end of Subsection \ref{ss:MD}, we said that $\frac 12 (1+T(\psi_A,\sigma_{\beta,A}))$ is the success probability of the optimal protocol for predicting whether a given state of subsystem $A$ is $\psi_A$ or $\sigma_{\beta,A}$ by performing a measurement on a single copy of the given state. The measurement in this protocol is in the eigenbasis of $\psi_A-\sigma_{\beta,A}$. Since both $\psi_A$ and $\sigma_{\beta,A}$ are Gaussian states, each of them can be written as a tensor product of single-mode states with an appropriate choice of modes. Furthermore, since the correlation matrices of $\psi_A$ and $\sigma_{\beta,A}$ commute, the same set of modes works for both states.  Thus, we can measure the observable  $\psi_A-\sigma_{\beta,A}$  by measuring $L$ single-mode operators and classically combining the results. Complete details of the measurement will be given in Subsection \ref{ss:met}.

\subsection{Thermalization}

We initialize the system in a product state with fermion number $n$ and let $\nu=n/N$. Since the ensemble of SYK$q$ Hamiltonians is invariant with respect to permutations of indices, we may assume without loss of generality that the initial state is
\begin{equation}
|\phi\rangle:=a_1^\dag a_2^\dag\cdots a_n^\dag|0\rangle,
\end{equation}
where $|0\rangle$ is the vacuum state with no fermions. Since $|\phi\rangle$ has a definite fermion number, $H$ and $H_\textnormal{SYK}$ generate the same dynamics in the sense that
\begin{equation} \label{eq:24}
e^{-iHt}|\phi\rangle=e^{-int}e^{-iH_\textnormal{SYK}\epsilon_1t}|\phi\rangle=e^{-int}e^{-i\epsilon_1H_\textnormal{SYK2}t-i\epsilon_1\epsilon_2H_\textnormal{SYK4}t}|\phi\rangle.
\end{equation}

Let $L,m$ be positive integers such that $Lm$ is a multiple of $N$. Let $A_1,A_2,\ldots,A_m$ be $m$ possibly overlapping subsystems, each of which has exactly $L$ fermionic modes. Suppose that each fermionic mode in the system is in exactly $Lm/N$ out of these $m$ subsystems. Let
\begin{equation}
\phi(t):=e^{-iHt}|\phi\rangle\langle\phi|e^{iHt},\quad\phi(t)_{A_j}:=\tr_{\bar A_j}\phi(t)
\end{equation}
be the state and its reduced density matrix at time $t$, respectively.

Let $\tau$ be sufficiently large\footnote{Conceptually, $\tau$ needs to be sufficiently large such that the effect of the SYK4 term in Eq. (\ref{eq:24}) is significant at most time $t\in[0,\tau]$. At a technical level, Theorem \ref{cor} follows from Theorem \ref{thm:et}. The proof of the latter theorem in Ref. \cite{HH23} approximates the infinite-time average $\lim_{\tau'\to\infty}\e_{t\in[0,\tau']}$ by the long-time average $\e_{t\in[0,\tau]}$. $\tau$ needs to be sufficiently large such that the approximation error is negligible.} and $t$ be uniformly distributed in the interval $[0,\tau]$. Recall the definition (\ref{eq:thr}) of $\sigma_{\beta,A}$, where $\beta$ is given by (\ref{eq:temp}).

\begin{theorem} [thermalization] \label{cor} 
For $1/2\ge\nu=\Omega(1)$, subsystems of size
\begin{equation} \label{eq:subu}
L=o(N/\ln N)
\end{equation}
thermalize in the sense that ($\poly(N)$ denotes a polynomial of sufficiently high degree in $N$)
\begin{multline} \label{eq:thm}
\Pr_h\left(\Pr_J\left(\Pr_{t\in[0,\tau]}\left(\frac1m\sum_{j=1}^m\|\phi(t)_{A_j}-\sigma_{\beta,A_j}\|_1^2=\frac{O(L\ln N)}N\right)=1-e^{-\Omega(N)}\right)=1\right)\\
\ge1-1/\poly(N).
\end{multline}
\end{theorem}

For any linear operator $B_j$ on subsystem $A_j$ with $\|B_j\|\le1$, Eq. (\ref{eq:tr}) implies that
\begin{equation}
\left|\tr\big(\phi(t)B_j\big)-\tr(\sigma_\beta B_j)\right|\le\|\phi(t)_{A_j}-\sigma_{\beta,A_j}\|_1.
\end{equation}
Thus, (\ref{eq:thm}) implies thermalization of physical properties measured on $o(N/\ln N)$ fermionic modes.

In contrast to Eq. (\ref{eq:subu}), reduced density matrices do not thermalize if the subsystem size is a constant fraction of the system size. Let $\e_{|A|=L}$ denote averaging over all subsystems of $L$ fermionic modes. There are ${N\choose L}$ such subsystems.

\begin{theorem} [failure of thermalization] \label{thm:nt}
Suppose that $1/2\ge\nu=\Omega(1)$. For $L=\Omega(N)$ and any $h,J,t$,
\begin{equation}
\e_{|A|=L}\|\phi(t)_A-\sigma_{\beta,A}\|_1=\Omega(1).
\end{equation}
\end{theorem}

\section{Proof sketches}

In this section, we give intuitive sketches of the proofs of our results.  Full calculations are deferred to Appendix~\ref{s:app}.

\subsection{Eigenstate thermalization} \label{ss:met}

For a density operator $\rho$, let $\langle B\rangle:=\tr(\rho B)$ denote the expectation value of an operator $B$. Let $\mathbf C$ be the correlation matrix with its elements given by
\begin{equation}
\mathbf C_{jk}:=\langle a_j^\dag a_k\rangle. \label{eq:C-def}
\end{equation}
It is easy to see that $\mathbf C$ is a Hermitian matrix of order $N$.

\begin{lemma} [\cite{BHK+22, HH23}] \label{l:haar}
$|\psi\rangle\sim\mathcal P_{N,n}$ means that $|\psi\rangle$ is a uniformly random Gaussian state with fermion number $n$ in the sense of Definition \ref{def:haar}.
\end{lemma}

\begin{definition} [uniformly random pure Gaussian state with definite fermion number] \label{def:haar}
A pure Gaussian state with fermion number $n$ is uniformly random if its correlation matrix is given by
\begin{equation}
\mathbf C=U^\dag\diag(\underbrace{1,1,\ldots,1}_\textnormal{$n$ ones},\underbrace{0,0,\ldots,0}_\textnormal{$N-n$ zeros})U,
\end{equation}
where $U$ is a unitary matrix chosen uniformly at random with respect to the Haar measure.
\end{definition}

Assume without loss of generality that the indices of the $L$ fermionic modes in subsystem $A$ are $1,2,\ldots,L$. Let $\mathbf C_A=(U_{n\times L})^\dag U_{n\times L}$ be the $L\times L$ upper left submatrix of $\mathbf C$, where $U_{n\times L}$ is the $n\times L$ upper left submatrix of $U$.

We can interpret $\mathbf C_A$ as the overlap of two projectors, as follows.  Let $P_n$ and $P_L$ be projectors of ranks $n$ and $L$, respectively, such that  $U_{n \times L} = P_n U P_L$. Then, $\mathbf C_A = P_L (U^\dag P_n U) P_L$.  We can view this as the overlap between a fixed projector $P_L$ and a random projector $U^\dag P_n U$.  

Let
\be
\mathbf C_A = V \diag(\lambda_1,\lambda_2,\ldots, \lambda_L) V^\dag
\ee
be the eigendecomposition of the Hermitian matrix $\mathbf C_A$, where $V$ is a unitary matrix of order $L$. Define a row of annihilation operators
\be
(b_1,b_2,\ldots,b_L) = (a_1,a_2,\dots,a_L)V.
\ee
Then, Eq. (\ref{eq:C-def}) implies that
\be \label{eq:cor}
\langle b_j^\dag b_k\rangle=\lambda_j \delta_{jk},
\ee
where $\delta_{jk}$ is the Kronecker delta. Since the reduced density matrix $\psi_A$ is a Gaussian state, it is fully determined by Eq. (\ref{eq:cor}) so that
\begin{equation} \label{eq:eiga}
  \psi_A= \prod_{j=1}^L \big( \lambda_j b_j^\dag b_j+(1-\lambda_j)b_jb_j^\dag\big).
\end{equation}
In the eigenbasis of $b_1^\dag b_1,b_2^\dag b_2,\ldots,b_L^\dag b_L$, $\psi_A$ has the matrix representation
\be
  \bigotimes_{j=1}^L\diag(\lambda_j,1-\lambda_j).
 \ee
In the same basis,
\begin{equation}  \label{eq:tha}
\sigma_{\beta,A}=\bigotimes_{j=1}^L\diag(\nu,1-\nu)
\end{equation}
is also a product state. Since the fidelity is multiplicative,
\begin{equation} \label{eq:f}
F(\psi_A,\sigma_{\beta,A})=\prod_{j=1}^L\left(\sqrt{\nu\lambda_j}+\sqrt{(1-\nu)(1-\lambda_j)}\right)^2.
\end{equation}

The optimal protocol described near the end of Subsection \ref{ss:MD} for predicting whether a given state of subsystem $A$ is $\psi_A$ or $\sigma_{\beta,A}$ proceeds as follows.  Measure the occupation numbers $b_1^\dag b_1,b_2^\dag b_2,\ldots, b_L^\dag b_L$ and let $m_1,m_2,\ldots,m_L$ be the corresponding measurement results. Each $m_j$ is a binary random variable with $\Pr(m_j=1)=\lambda_j$ or $\Pr(m_j=1)=\nu$ if the given state is $\psi_A$ or $\sigma_{\beta,A}$, respectively. We predict $\psi_A$ if
  \be\prod_{j=1}^L(\lambda_j/\nu)^{m_j}\ge1 \ee
  and predict $\sigma_{\beta,A}$ otherwise.
  This can be thought of as a likelihood-ratio test, in which we predict the state that makes our measurement outcomes more likely.  It is also the optimal Helstrom measurement with success probability $\frac12(1+T(\psi_A,\sigma_{\beta,A}))$.

\begin{proof} [Proof of Theorem \ref{thm:nontd}]
For $L>n$, $\mathbf C_A$ is singular, and the multiplicity of the eigenvalue $0$ is at least $L-n$. Using Eqs. (\ref{eq:eiga}), (\ref{eq:tha}) and since the fidelity (trace distance) is non-decreasing (non-increasing) under partial trace,
\begin{gather}
F(\psi_A,\sigma_{\beta,A})\le F\left(\bigotimes_{j=1}^{L-n}\diag(0,1),\bigotimes_{j=1}^{L-n}\diag(\nu,1-\nu)\right)=(1-\nu)^{L-n},\\
T(\psi_A,\sigma_{\beta,A})\ge T\left(\bigotimes_{j=1}^{L-n}\diag(0,1),\bigotimes_{j=1}^{L-n}\diag(\nu,1-\nu)\right)=1-(1-\nu)^{L-n}.
\end{gather}
\end{proof}

For $L\le\min\{n,N-n\}$, the joint probability distribution of $\lambda_1,\lambda_2,\ldots,\lambda_L$ is the Jacobi unitary ensemble with parameters $r=N-n-L$ and $s=n-L$ \cite{Rou07, LCB18, BHK+22}.

\begin{definition} [Jacobi unitary ensemble]
The probability density function of the Jacobi unitary ensemble with parameters $r,s>-1$ is
\begin{equation}
\mathcal J(\lambda_1,\lambda_2,\ldots,\lambda_L)\propto{\prod_{1\le j<k\le L}(\lambda_j-\lambda_k)^2}\prod_{j=1}^L(1-\lambda_j)^r\lambda_j^s,\quad0\le\lambda_j\le1.
\end{equation}
\end{definition}

We explain why the ETH holds and fails with high probability for $L=o(\sqrt N)$ and $L=\Omega(\sqrt N)$, respectively. To this end, consider the limit $L\to\infty$ with $L=o(N)$. In this case, Theorem 2.1 in Ref. \cite{Nag14} says that the empirical distribution of
\begin{equation}
\sqrt{\frac{N}{\nu(1-\nu)L}}(\lambda_j-\nu)
\end{equation}
converges weakly to the semicircle distribution with radius $2$ almost surely. Thus,
\begin{equation}
\lambda_j=\nu \pm \Theta(\sqrt{L/N})
\end{equation}
for almost all $j$ with high probability. Since $L=o(N)$, by Taylor expansion,
\begin{equation}
\sqrt{\nu\lambda_j}+\sqrt{(1-\nu)(1-\lambda_j)}=1-\Theta(\lambda_j-\nu)^2.
\end{equation}
Substituting this into Eq. (\ref{eq:f}),
\begin{equation}
F(\psi_A,\sigma_{\beta,A})=\prod_{j=1}^L\big(1-\Theta(L/N)\big)=
\begin{cases}
1-o(1), & L=o(\sqrt N),\\
1-\Omega(1), & L=\Omega(\sqrt N).
\end{cases}
\end{equation}

Proofs of the probabilistic bounds in Theorems \ref{thm:ethv}, \ref{thm:non} without assuming $L\to\infty$ or $L=o(N)$ are given in Appendix \ref{ss:et}.  There we do not use Theorem 2.1 in Ref. \cite{Nag14} but instead rely on a recent concentration result \cite{HH22} for the second moment of the Jacobi ensemble.

\subsection{Thermalization}

\paragraph{Proof sketch of Theorem \ref{cor}.}In a previous paper \cite{HH23}, we proved entanglement thermalization: With high probability, the von Neumann entropies of $\phi(t)_{A_j}$ and $\sigma_{\beta,A_j}$ are equal to leading order in $L$. So are the free energies of $\phi(t)_{A_j}$ and $\sigma_{\beta,A_j}$. Since the thermal state minimizes the free energy \cite{Weh78}, the free energy of $\phi(t)_{A_j}$ is only slightly higher than the minimum. Pinsker's inequality \cite{Aud14} implies that any state of low free energy is close to a thermal state in trace distance.

Condition (\ref{eq:subu}) ensures that with high probability, the free energy of $\phi(t)_{A_j}$ is sufficiently low so that the trace distance between $\phi(t)_{A_j}$ and $\sigma_{\beta,A_j}$ from the above analysis is $o(1)$.

\paragraph{Proof sketch of Theorem \ref{thm:nt}.}$\phi(t)$ for any $t\in\mathbb R$ has a definite fermion number, but $\sigma_\beta$ does not. This difference has its footprint in the reduced density matrices. If the fermion number operator (\ref{eq:qa}) is measured on a random subsystem $A$ of $L$ fermionic modes, we obtain a probability distribution on the integers $0,1,2,\ldots,L$. For $\phi(t)$, the distribution is hypergeometric corresponding to drawing $L$ balls without replacement from a pool of $n$ white and $N-n$ black balls. For $\sigma_\beta$, the distribution is binomial corresponding to drawing $L$ balls with replacement from the same pool.

\paragraph{Note added.}Very recently, we became aware of related work by Yu, Gong, and Cirac \cite{YGC23}. They studied the entanglement of random Gaussian states with definite fermion number.  Their Theorem 1 is conceptually similar to our Lemma \ref{lem:M-conc} but gives different bounds and is proved using different methods.

\section*{Acknowledgments}

This material is based upon work supported by the U.S. Department of Energy, Office of Science, National Quantum Information Science Research Centers, Quantum Systems Accelerator.  AWH was also supported by NSF grants CCF-1729369 and PHY-1818914 and NTT (Grant AGMT DTD 9/24/20).

\appendix

\section{Proofs} \label{s:app}

\subsection{Eigenstate thermalization} \label{ss:et}

\begin{lemma} \label{l:small}
For any $x,y\in\mathbb R^+$ such that $x+y\le1$,
\begin{equation}
1-(x-y)^2/\max\{x,y\}\le\big(\sqrt{xy}+\sqrt{(1-x)(1-y)}\big)^2\le1-(x-y)^2.
\end{equation}
\end{lemma}

\begin{proof}
Let
\begin{equation}
x=\cos^2\frac{\theta+\alpha}2,\quad y=\cos^2\frac{\theta-\alpha}2,\quad0\le\theta\pm\alpha\le\pi
\end{equation}
so that
\begin{equation} \label{eq:40}
\big(\sqrt{xy}+\sqrt{(1-x)(1-y)}\big)^2=\cos^2\alpha=1-\sin^2\alpha,\quad(x-y)^2=\sin^2\theta\cdot\sin^2\alpha.
\end{equation}
The condition $x+y\le1$ implies that $\theta\ge\pi/2$ so that
\begin{equation} \label{eq:41}
\max\{x,y\}\le\max_{\alpha\in[0,\pi-\theta]}\cos^2\frac{\theta-\alpha}2=\sin^2\theta\le1.
\end{equation}
We complete the proof by combining (\ref{eq:40}), (\ref{eq:41}).
\end{proof}

Let
\begin{equation}
M:=\sum_{j=1}^L(\lambda_j-\nu)^2
\end{equation}
be the shifted second moment of the Jacobi ensemble.

\begin{lemma} [\cite{HH22}] \label{lem:M-conc}
For any $\delta>0$,
\begin{equation}
\Pr\big(|M-\nu(1-\nu)L^2/N|>\delta\big)=O(e^{-\Omega(N\delta)\min\{N\delta/L^2,1\}}).
\end{equation}
\end{lemma}

\begin{proof} [Proof of Theorem \ref{thm:ethv}]
Using Eq. (\ref{eq:f}) and Lemma \ref{l:small},
\begin{equation}
F(\psi_A,\sigma_{\beta,A})\ge\prod_{j=1}^L\max\{1-(\lambda_j-\nu)^2/\nu,0\}\ge1-M/\nu.
\end{equation}
Then, Eq. (\ref{eq:ethf}) follows from Lemma \ref{lem:M-conc}. Equation (\ref{eq:etht}) follows from (\ref{eq:rel}) and Eq. (\ref{eq:ethf}).
\end{proof}

\begin{proof} [Proof of Theorem \ref{thm:non}]
Using Eq. (\ref{eq:f}) and Lemma \ref{l:small},
\begin{equation} \label{eq:fe}
F(\psi_A,\sigma_{\beta,A})\le\prod_{j=1}^L\big(1-(\lambda_j-\nu)^2\big)\le e^{-M}.
\end{equation}
For $L\le n$, Eq. (\ref{eq:nons}) follows from (\ref{eq:fe}) and Lemma \ref{lem:M-conc}. 

For $L>n=\Omega(N)$, let $A'$ be an (arbitrary) subsystem of $n$ fermionic modes in $A$. We have proved that
\begin{equation}
\Pr_{|\psi\rangle\sim\mathcal P_{N,n}}\big(F(\psi_{A'},\sigma_{\beta,A'})=e^{-\Omega(N)}\big)=1-e^{-\Omega(N^2)}.
\end{equation}
Equation (\ref{eq:nons}) follows the fact that the fidelity is non-decreasing under partial trace.

Equation (\ref{eq:nont}) follows from (\ref{eq:rel}) and Eq. (\ref{eq:nons}).
\end{proof}

\subsection{Proof of Theorem \ref{cor}}

Let
\begin{equation}
S(\rho):=-\tr(\rho\ln\rho)
\end{equation}
be the von Neumann entropy of a density operator and
\begin{equation}
H_b(\nu):=-\nu\ln\nu-(1-\nu)\ln(1-\nu)
\end{equation}
be the binary entropy function.

\begin{theorem} [\cite{HH23}] \label{thm:et}
Suppose that $1/2\ge\nu=\Omega(1)$. For $L\le N/10$,
\begin{multline}
\Pr_h\left(\Pr_J\left(\Pr_{t\in[0,\tau]}\left(\frac1m\sum_{j=1}^mS(\phi(t)_{A_j})\ge H_b(\nu)L-\frac{O(L\ln N)}N\right)=1-e^{-\Omega(N)}\right)=1\right)\\
\ge1-1/\poly(N).
\end{multline}
\end{theorem}

Recall the definition (\ref{eq:qa}) of $Q_A$. It is easy to see that
\begin{equation}
  \frac1m\sum_{j=1}^mQ_{A_j}=\frac{LQ}{N}
  \implies\frac1m\sum_{j=1}^m\tr\big(\phi(t)_{A_j}Q_{A_j}\big)=\nu L.
\end{equation}
Using Pinsker's inequality \cite{Aud14} between the trace distance and quantum relative entropy,
\begin{multline}
  \frac 12 \|\phi(t)_A-\sigma_{\beta,A}\|_1^2
  \le S\big(\phi(t)_A\|\sigma_{\beta,A}\big):=-\tr\big(\phi(t)_A\ln\sigma_{\beta,A}\big)-S(\phi(t)_A)\\
=\beta\tr\big(\phi(t)_AQ_A\big)+\ln\tr(e^{-\beta Q_A})-S(\phi(t)_A)
\end{multline}
so that
\begin{equation}
\frac1{2m}\sum_{j=1}^m\|\phi(t)_{A_j}-\sigma_{\beta,A_j}\|_1^2\le H_b(\nu)L-\frac1m\sum_{j=1}^mS(\phi(t)_{A_j}).
\end{equation}
Theorem \ref{cor} follows from this inequality and Theorem \ref{thm:et}.

\subsection{Proof of Theorem \ref{thm:nt}}

Let $i_1<i_2<\cdots<i_L$ be the indices of the $L$ fermionic modes in $A$. Define
\begin{equation} \label{eq:PAdef}
  P_A^{>l}=\sum_{\substack{(n_1,n_2,\ldots,n_L)\in\{0,1\}^{\times L} \\ \sum_{j=1}^Ln_j>l}}
  \prod_{j=1}^L\big(n_j+(1-2n_j)a_{i_j}a_{i_j}^\dag\big).
\end{equation}
Let $[N] := \{1,\ldots,N\}$ and
\begin{equation}
\binom{[N]}{n} := \{R \subseteq [N] : |R|=n\}
\end{equation}
be the set of size-$n$ subsets of $[N]$.  Let $\{|\phi_R\rangle\}_{R \in \binom{[N]}{n}}$ be the complete set of computational basis states with $n$ fermions, where
\be
\ket{\phi_R} := \left(\prod_{j\in R} a_j^\dag\right) \ket{0}.
\ee

By construction, $P_A^{>l}=(P_A^{>l})^2$ is a projector such that 
\begin{equation}
P_A^{>l}|\phi_R\rangle=|\phi_R\rangle~\textnormal{or}~0
\end{equation}
if $\ket{\phi_R}$ contains $>l$ or $\leq l$ fermions in $A$, respectively. Hence, 
\begin{multline}
  \e_{|A|=L}\|P_A^{>l}|\phi_R\rangle\|
  ={N\choose L}^{-1}\big|\{R'\subseteq[N]:|R'|=L~\textnormal{and}~|R\cap R'| > l\}\big|\\
  ={N\choose L}^{-1}\sum_{j>l}{n\choose j}{N-n\choose L-j}.
\end{multline}
The time-evolved state can be expanded as
\begin{equation}
e^{-iHt}|\phi\rangle=\sum_{R \in \binom{[N]}{n}} c_R(t)|\phi_R\rangle
\end{equation}
so that
\begin{multline} \label{eq:thg}
  \e_{|A|=L}\tr\big(\phi(t)_AP_A^{>l}\big)=\e_{|A|=L}\|P_A^{>l}e^{-iHt}|\phi\rangle\|^2=
  \e_{|A|=L} \sum_{R \in \binom{[N]}{n}} |c_R(t)|^2\|P_A^{>l}|\phi_R\rangle\|\\
  =\sum_{R \in \binom{[N]}{n}} |c_R(t)|^2
  \e_{|A|=L}\|P_A^{>l}|\phi_R\rangle\|
  ={N\choose L}^{-1}\sum_{j>l}{n\choose j}{N-n\choose L-j},\quad\forall t\in\mathbb R.
\end{multline}
Equations (\ref{eq:thr}), (\ref{eq:temp}) imply that
\begin{equation} \label{eq:tb}
\tr(\sigma_{\beta,A}P_A^{>l})=\frac1{\tr(e^{-\beta Q_A})}\sum_{j>l}{L\choose j}e^{-\beta j}=\sum_{j>l}{L\choose j}\nu^j(1-\nu)^{L-j}.
\end{equation}
Equations (\ref{eq:thg}) and (\ref{eq:tb}) are the tails of the hypergeometric and binomial distributions, respectively.
The distributions have the same mean $\nu L$ but different variances: $\nu(1-\nu)L\frac{N-L}{N-1}$  for the hypergeometric distribution and $\nu(1-\nu)L$ for the binomial distribution.  Furthermore, both distributions are well approximated by Gaussians matching those moments. Thus, their tails are distinguishable in the sense of (\ref{eq:dst}) for $l=\nu L+\Theta(\sqrt L)$.

The total variation distance between the hypergeometric and binomial distributions has been studied in the context of de Finetti theorem. The theorem states that for a permutation-invariant probability distribution of $N$ random variables, the marginal distribution of $L \ll N$ variables is close to a mixture of distributions, each of which represents $L$ independent and identically distributed random variables \cite{DF80}.
From either the Gaussian approximation or (in the $\nu=1/2$ case) Theorem 35 and Lemmas 45, 46 in Ref. \cite{DF80}, we have
\begin{equation} \label{eq:dst}
0<\e_{|A|=L}\tr(\sigma_{\beta,A}P_A^{>l})-\e_{|A|=L}\tr\big(\phi(t)_AP_A^{>l}\big)=\Omega(1)
\end{equation}
for $L=\Omega(N)$. We complete the proof by noting that
\begin{equation}
\e_{|A|=L}\|\phi(t)_A-\sigma_{\beta,A}\|_1\ge\e_{|A|=L}\left|\tr\big(\phi(t)_AP_A^{>l}-\sigma_{\beta,A}P_A^{>l}\big)\right|\ge\left|\e_{|A|=L}\tr\big(\phi(t)_AP_A^{>l}-\sigma_{\beta,A}P_A^{>l}\big)\right|.
\end{equation}

\printbibliography

\end{document}